\documentclass[abstract=on,notitlepage,superscriptaddress,nofootinbib,aps,showpacs,twocolumn,prd]{revtex4}
\usepackage{srcltx}
\usepackage{epsf}
\usepackage{amsfonts}
\usepackage{amsmath}
\usepackage{mathrsfs}
\usepackage{epsfig}
\usepackage{enumerate}
\usepackage{lipsum}
\usepackage{graphicx}
\usepackage{epsf}
\usepackage{subfig}
\usepackage{color}
\usepackage{caption}
\usepackage{subfig}
\usepackage{epstopdf}
\usepackage{float}
\usepackage{dcolumn}
\usepackage{hyperref}
\usepackage{amsthm}

\newtheorem{prop}{Proposition}

\newtheorem{thm}{Theorem}

\newcommand{\be}{\begin{equation}}
\newcommand{\ee}{\end{equation}} 
\newcommand{\eei}{\end{equation}\indent\indent}
\newcommand{\bc}{\begin{center}}
\newcommand{\ec}{\end{center}}
\newcommand{\ber}{\begin{eqnarray*}}
\newcommand{\ear}{\end{eqnarray*}}
\newcommand{\ba}{\begin{array}}
\newcommand{\ea}{\end{array}}

\newcommand{\na}{\nabla}

\newcommand{\vv}{{\cal V}}

\newcommand{\bea}{\begin{eqnarray}}
\newcommand{\eea}{\end{eqnarray}}

\newcommand{\ei}{\end{itemize}}

\newcommand{\bra}[1]{\left(#1\right)}
\newcommand{\bras}[1]{\left[#1\right]}

\newcommand{\Lietwo}{{\cal L}}

\newcommand{\reff}[1]{(\ref{#1})}

\def\case#1/#2{\textstyle\frac{#1}{#2} }

\def\fp{f_{,R}}
\def\fpp{f_{,RR}}

 \def\prd{Phys.\ Rev.\ D }

 \def\plb{{\em Phys. Lett.\/} B\/}

\def\cqg{{\em Class. Quantum Grav.\/} }

\def\aph{{\em Ann. Phys. (NY)\/} }

 {\begin{list}{}%
         {\setlength{\leftmargin}{#1} %
         \setlength{\rightmargin}{#2}%
         }%
         \item[\textit{}]%
 }
 {\end{list}}

\makeatletter

\begin{document}
\sloppy

\title{Buchdahl-Bondi limit in modified gravity: Packing extra effective mass in relativistic compact stars}

\author{Rituparno Goswami}
\email{Goswami@ukzn.ac.za}
\affiliation{Astrophysics \& Cosmology Research Unit,
School of Mathematics Statistics and Computer Science,
University of KwaZulu-Natal,
Private Bag X54001, Durban 4000, South Africa.}

\author{Sunil D. Maharaj}
\email{Maharaj@ukzn.ac.za}
\affiliation{Astrophysics \& Cosmology Research Unit,
School of Mathematics Statistics and Computer Science,
University of KwaZulu-Natal,
Private Bag X54001, Durban 4000, South Africa.}

\author{Anne Marie Nzioki}
\email{anne.nzioki@gmail.com}
\affiliation{Astrophysics \& Cosmology Research Unit,
School of Mathematics Statistics and Computer Science,
University of KwaZulu-Natal,
Private Bag X54001, Durban 4000, South Africa.}

\begin{abstract}
We generalise the Buchdahl-Bondi limit for the case of static, spherically symmetric, relativistic compact stars immersed in Schwarzschild vacuum in $f(R)$-theory of gravity, subject to very generic regularity, thermodynamic stability and matching conditions. Similar to the case of general relativity, our result is model independent and remains true for any physically realistic equation of state of standard stellar matter. We show that an extra-massive stable star can exist in these theories, with surface redshift larger than 2, which is forbidden in general relativity. This result gives a novel and interesting observational test for validity or otherwise of general relativity and also provides a possible solution to the dark matter problem.
\end{abstract}
\maketitle
\nopagebreak

\section{Introduction}
In general relativity (GR), differentiable properties and regularity conditions of Einstein field equations lead to certain interesting bounds on stellar structures \cite{chandra}.
One of such fascinating upper bound, is the Buchdahl-Bondi bound \cite{buch1,buch2,bondi,islam,wald,stephani} for the static, spherically symmetric compact stars immersed in Schwarzschild vacuum. This bound states that the mass to radius ratio $2M/r_b$ of any regular and thermodynamically stable perfect fluid star must be less than $8/9$. We note that this is a stricter upper bound than the Schwarzschild static limit $2M/r_b=1$. In other words, even though the mass of the star lies within the untrapped region, we cannot have a stable static stellar configuration for $(8/9)r_b\le2M<r_b$. The interesting points about this general result can be summarised as follows:
\begin{enumerate}[(a)]
\item This bound can be proved using minimal technical conditions, namely the regularity and smoothness of metric functions in the stellar interior, matching conditions at the boundary of the star where the spacetime is smoothly matched to a Schwarzschild exterior. Also for thermodynamic stability, we need to impose the condition that the average density is a non-increasing function of the radial coordinate.
\item This result is model independent. Hence it is true for any physically realistic equation of state for the stellar matter.
\item A direct consequence of this bound is that the gravitational redshift $z$ at the stellar surface is bounded from above ($z\le2$) \cite{wald,stephani}.
\end{enumerate}
Several variations/modifications of this limit have been found since, by altering the conditions mentioned above \cite{visser, andreasson, lemos, zarro} or by considering anisotropic stars (see for example \cite{boehmer,ivanov,lake} and the references therein). 

Now the key question that arises here is, though this result is quite generic within GR, how does it change if the theory of gravity is modified. This question is important  even though GR still remains as the successful theory of gravity so far, alternative theories of gravity do exist. These theories are motivated by the ambiguous nature of dark energy in cosmology which is responsible for the observed late time accelerated expansion of the universe. Another motivation, of course, is the  unknown nature of dark matter which dominates the matter budget of the galaxies that can be gravitationally detected by the galaxy rotation curves. 

To eliminate the need for the so called exotic dark sectors in the universe, an alternative possibility is to conjecture that GR is an effective local  theory of a more general theory on universal scale. Among the numerous modifications of GR that naturally provide a late time cosmological accelerated expansion, without the need for the presence of dark fluids, is $f(R)$-gravity. This theory is based on a gravitational action that contains an arbitrary but well defined function of the Ricci scalar $R$ \cite{DEfR,fr,fr2,fr1,scm}. By expanding the function $f$ around GR, we can easily see that the higher order curvature terms naturally admit a phase of accelerated expansion both in the early universe as an inflationary phase \cite{star80}, and also in the late times after passing through a matter dominated decelerating expansion phase\cite{shosho}. This theory of gravity essentially contain an additional scalar degree of freedom which is the Ricci scalar. Hence, this class of theories can be considered as a natural extension to GR, and contrary to other theories that have the square of Ricci or Riemann tensors in the action, the Ostrogradski instability can be eliminated, despite the equations of motion being fourth-order in the metric components.

\subsection{The key question}

If we assume the framework of $f(R)$- gravity as the modification to GR on larger scales and strong gravity regimes, then the natural question one would ask is:\\

{\bf Question:} {\em Is it possible to have more massive but stable compact stars in $f(R)$-theories of gravity? In other words, can the mass to radius ratio of a compact star lie in the forbidden region in GR, thus making the surface redshift of the star larger than 2?}\\

The importance of this question is two-fold. If we can prove that the mass to radius ratio of a star can be within the forbidden region in these theories, then this gives a direct experimental test for the validity or otherwise of GR in strong gravity regime. Secondly the extra effective mass that can be packed in compact stars can account, to some extent, for the dark matter in the galaxies, that only have gravitational signatures.

\subsection{In this paper}

In this paper we discuss and try to answer the above question, for stable static relativistic compact stars in $f(R)$ gravity, immersed in the Schwarzschild vacuum. The existence of such stellar solutions is already shown in many earlier works (see for example \cite{3GR,ssfr} and the references therein). Apart from the regularity, thermodynamic stability and matching conditions that are used to prove the Buchdahl-Bondi limit in GR, we now have to deal with other extra conditions:
\begin{enumerate}
\item The existence of the Schwarzschild spacetime as a solution to the field equations, gives certain conditions on the function $f$. Since Schwarzschild spacetime is experimentally well verified around stars (e.g. the solar system), we will only consider those class of the function $f$ that admits a Schwarzschild solution.
\item For higher order theories such as $f(R)$, we have extra matching conditions on the surface of the star \cite{Deruelle, Clifton, Senovilla}, apart from the usual Israel-Darmois \cite{Israel, Darmois} conditions in GR.
\item Since the Ricci scalar is a dynamic degree of freedom in this theory, we must impose certain physically reasonable conditions on it in the interior of the star (such as monotonicity) to get a physically viable result.
\end{enumerate}
Taking all these extra conditions into account, in this paper we transparently demonstrate the modification of Buchdahl-Bondi limit in $f(R)$-theories, and show that a stable star can exist in the otherwise forbidden region in GR, with surface redshift larger than 2. This result is absolutely generic since it is model independent as in GR. It is true for any equation of state for standard stellar matter and also for any function $f$ that satisfies the conditions stated above. \\

Unless otherwise specified, throughout this paper we use units which fix the speed of light and the
gravitational constant via $8\pi G = c^4 = 1$, and the metric signature is $+2$.

\section{Field equations for relativistic static stars in $f(R)$-gravity}
To obtain the field equations in $f(R)$-gravity we begin with the modification to the Einstein-Hilbert action by generalising the Lagrangian  
so that the Ricci scalar $R$ is replaced by a function $f(R)$. The modified action, therefore, is given by
\be
\label{action}
{\cal S}= \frac12 \int dV \bras{\sqrt{-g}\,f(R)+ 2\,\Lietwo_{M}(g_{ab}, \psi) } ~,
\ee
where $ \Lietwo_{M} $ is the Lagrangian density of the matter fields $\psi$,  $g$ is
the determinant of the metric tensor $g_{ab}$ $(a,b=0,
1, 2, 3)$, $R$ is the scalar curvature, and $f(R)$ is the real valued and well behaved 
function defining the theory under consideration. Varying the action \reff{action} with respect to the metric over a 4-volume
yields the field equations
\be
\label{field}
G_{ab} \, \fp -  \frac12 g_{ab}\,(f-R \,  \fp)
 -\na_{a}\na_{b}\fp +g_{ab} \, \Box  \fp =   T^{M}_{ab}~,
\ee
where $\fp =d f(R)/dR $, $\Box \equiv \na_{c}\na^{c}$, $G_{ab}$ is the Einstein tensor and
$T^{M}_{ab} $ is the matter energy momentum tensor defined by 
\be
 \label{EMT}
T^{M}_{ab}=- \frac{2}{\sqrt{-g}} \, \frac{\delta \Lietwo_{M}} {\delta g^{ab} } ~.
\ee
We can easily see that in the special case of  $f(R) = R$, the field equations \reff{field} reduce to the
standard Einstein field equations. These
theories are also  known as fourth-order gravity, since the
term $( g_{ab} \Box - \nabla_a \nabla_b )\fp $ has fourth-order
derivatives of the components of the metric tensor.  

We can map the fourth order field equatiions \reff{field} to the effective Einstein equations by
\be
\label{field2} 
G_{ab} = \bra{R_{ab}-\frac12 g_{ab} \, R} 
= \tilde{T}^{M}_{ab} + T^{R}_{ab} = T_{ab}~,
\ee
where we define $T_{ab}$ as the total energy momentum tensor comprising of an effective matter contribution
\be
\tilde{T}^{M}_{ab} = \frac{T^{M}_{ab}}{ \fp}\;,
\label{matterEMT} 
\ee 
and the contribution from the scalar curvature of the spacetime as
\be
T^{R}_{ab} = \frac{1}{ \fp} \bras{\frac12g_{ab} \,(f-R\, \fp) 
+\na_{a}\na_{b} \fp- g_{ab}\,\Box \fp}~. 
\label{curvatureEMT}
\ee
Writing the field equations in the form (\ref{field2}) has an extra advantage. We can consider the dynamics of the fourth order gravity on the same footing as that of general relativity with two matter fields, the first is sourced by the standard matter and second by the scalar curvature and its derivatives. Hence to analyse the properties of these field equations, we can implement the well understood  techniques from general relativity. 

\subsection{Spherical static star in $f(R)$-gravity}

To study a static spherically symmetric star, immersed in the vacuum Schwarzschild exterior, we consider a spherically symmetric static metric in the stellar interior as
\be 
\label{genmetric} 
ds^2 = - c^2(r) dt^2 + \frac{dr^2}{\left(1-\frac{2m(r)}{r}\right)} + r^2\, d\Omega^2 ~, 
\ee
where $ d\Omega^2  = d\theta^2 +\sin^2\theta \, d \phi^2 $, $c(r)$ and $m(r)$ are well defined (at least $C^2$) functions of the radial coordinate.
We consider the standard matter to be a spherically symmetric perfect fluid with no pressure anisotropy. However we cannot impose the same conditions on the ``curvature" fluid, as in general it is anisotropic. Hence in the above coordinate system the total effective energy momentum tensor is of the following form:
\be
T_{ab}={\rm diag}\left[ -\mu(r),p_r(r),p_\theta(r),p_\theta(r)\right]\;,
\label{em1}
\ee
where we have set
\bea\label{em2}
\mu(r)&=&\frac{\mu^M(r)}{\fp}+\mu^R(r)\;,\\
p_r(r)&=&\frac{p^M(r)}{\fp}+p_r^R(r)\;,\label{em3}\\
p_\theta(r)&=&\frac{p^M(r)}{\fp}+p_\theta^R(r)\;.\label{em4}
\eea
Using the definition (\ref{curvatureEMT}) and the metric (\ref{genmetric}) we can easily express the pressure anisotropy of the curvature fluid in the form
\bea
\label{anis}
p_\theta^R-p_r^R&=&\frac{1}{\fp}\left[\left(\frac{m'}{r}-\frac{m}{r^2}\right)R'\fpp\right.\nonumber\\
&&\left.-\left(1-\frac{2m}{r}\right)\left((R'\fpp)'+\fpp\frac{R'}{r}\right)\right]\;, \nonumber\\
\eea
where $'$ denotes differentiation with respect to the coordinate $r$.

\subsection{The effective mass}
The important question that arises here is, what would be the mass of the star, as measured by a faraway observer? Since any realistic astrophysical stars are immersed in the Schwarzschild vacuum exterior, the total effective mass of the star will be equal to the Schwarzschild mass in the exterior spacetime. Now from the $G^0_0=T^0_0$ field equation, we get 
\be 
\mu=\frac{2m'}{r^2}\;\Rightarrow\;2m(r)=\int\limits_0^r \mu(x)x^2 dx\;.
\label{G00}
\ee
Hence we can interpret the function $m(r)$ as the effective mass, which is generated by the standard matter and also the curvature terms,  enclosed within the shell of physical radius `$r$'. If $r=r_b$ denotes the boundary of the star, then $m(r_b)=M$ will denote the total effective mass of the star, and in the exterior vacuum spacetime this will be the Schwarzschild mass. Thus we immediately see that the total effective mass of the star in a fourth order theory is different from the mass which is generated purely by the standard matter inside the star. The other field equations are then used to relate this effective mass with the pressure terms.
The $G^1_1=T^1_1$ field equation gives
\be
p_r=\frac{2c'}{rc}\left(1-\frac{2m}{r}\right)-\frac{2m}{r^3}\;,
\label{G11}
\ee
while from the doubly contracted Bianchi identity $\nabla_aG^a_r=0=\nabla_aT^a_r$, we get 
\be
(cp_r)'+c'\mu=\frac{2c}{r}(p_\theta-p_r)\;.
\label{conserv}
\ee
Substituting equations (\ref{G00}) and (\ref{G11}) in the LHS of (\ref{conserv}), we have after simplification
\be\label{master1}
\sqrt{1-\frac{2m}{r}}\frac{d}{dr}\left[\frac{1}{r}\sqrt{1-\frac{2m}{r}}c'\right]= c\left[\left(\frac{m}{r^3}\right)'+\frac{p_\theta-p_r}{r}\right]\;.
\ee

\subsection{Problem of an anisotropic star? Not quite.}

The equation (\ref{master1}) is used extensively in the literature to study variations of the Buchdahl-Bondi limit in case of stars having pressure anisotropy (see for example \cite{boehmer,ivanov,lake} and the references therein). However this problem (as we shall see in this section), is not just another problem of an anisotropic star, where the anisotropy parameter is a free parameter. In this case, the pressure anisotropy is a direct consequence of  packaging all the higher derivative terms in the field equations as the energy momentum tensor of a `fictitious' curvature fluid. Hence, the anisotropy parameter depends on the Ricci scalar and it's derivatives, which in turn depend on the metric function, thereby creating a feedback loop. 

To see this feedback effect more transparently, substitute the curvature terms  in (\ref{master1}). Since we have already considered the standard matter to be a perfect fluid, we can immediately see from equations (\ref{em3}) and (\ref{em4}) that 
$p_\theta-p_r=p_\theta^R-p_r^R$. Then using equation (\ref{anis}) in the above expression, we get after some simplification
\bea
\frac{c\fp}{\sqrt{1-\frac{2m}{r}}}\frac{d}{dr}\left(\frac{m}{r^3}\right)&=&\fp\frac{d}{dr}\left[\frac{1}{r}\sqrt{1-\frac{2m}{r}}c'\right]\nonumber\\
&&+c\frac{d}{dr}\left[\frac{1}{r}\sqrt{1-\frac{2m}{r}}(\fp)'\right].
\label{master2}
\eea
The term $m/r^3$ in the LHS of the above equation denotes the average effective density of the star at a radius $r$. It is interesting to note that the metric function $c(r)$ (which depends on the gravitational potential) and the function $\fp$ play a symmetric role in determining the gradient of the effective average density.  When $\fp=1$ the above relation reduces to the well known relation of general relativity \cite{stephani}, which is used to calculate the Buchdahl-Bondi limit.

\section{Conditions on the functions }

In order to have a stable and smooth stellar structure which can be matched smoothly to the Schwarzschild vacuum exterior at the stellar boundary, we have to impose a set of boundary conditions on the interior metric functions as well as placing restrictions on the thermodynamic quantities of the star. Furthermore, to avoid the presence of ghosts or tachyons and to have the Schwarzschild solution as a viable solution of the theory, there are certain conditions on the function $f(R)$. In this section we list all these conditions, which will be used to calculate the bound on the star mass as dictated by the thermodynamic stability.

\subsection{Conditions of the function $f(R)$}

To ensure the attractive nature of gravity, that is the absence of ghost modes and tachyonic fields we must have in the stellar interior
\be
\fp>0\;,\;\; \fpp\ge0\;.
\label{cond1}
\ee
Furthermore, as proved in \cite{amn}, to have the Schwarzschild spacetime as a solution of the theory, the function $f(R)$ must be at least of class $C^3$ with
\be
f(0)=0\;,\;\; \fp(0)\ne 0\;.
\label{cond2}
\ee
Once these conditions are fulfilled, we can have a star with no ghost modes or tachyonic instabilities (which are unphysical as they would destroy stable stellar structures) and furthermore the star can be matched to a vacuum Schwarzschild exterior.

\subsection{Matching conditions}

We know that all astrophysical objects are immersed in vacuum or almost vacuum spacetime 
(like any star within a stellar system), and hence the exterior spacetime around 
a spherically symmetric star is well described by the Schwarzschild geometry. 
We match two spacetimes $\vv^{\pm}$ across the boundary surface $\Sigma$, which in this case will be $r=r_b$. 
The junction surface must be the same in $\vv^{+}$ and
$\vv^{-}$, which implies continuity of both the metric (the first fundamental form) and the extrinsic curvature (the second fundamental form)
across $\Sigma$ as in GR \cite{Israel, Darmois}. Moreover, in $f(R)$-theories of gravity, continuity of the Ricci scalar across the boundary surface and 
continuity of its normal derivative are also required \cite{Deruelle, Clifton, Senovilla}. 
Matching of the first fundamental form requires that
\be
c^2(r_b)= \left(1-\frac{2M}{r_b}\right)\;,
\label{cond3}
\ee
while matching the second fundamental form dictates that the total radial pressure at the surface of the star must vanish ($p_r(r_b)=0$). Therefore using (\ref{G11}) we get
\be
c'(r_b)=\frac{1}{r_b^2}\frac{M}{\sqrt{1-\frac{2M}{r_b}}}\;.
\label{cond4}
\ee
The extra matching conditions for $f(R)$ gravity give
\be
R(r_b)=0=R'(r_b)\;.
\label{cond5}
\ee

\subsection{Regularity Conditions}

The regularity conditions for the interior spacetime dictate that all the metric and the thermodynamic functions should be smooth (at least $C^2$) in the interior of the star. Hence at the centre of the star, the radial derivatives of all the metric and the thermodynamic functions should vanish. Therefore we must have
\be
c'(0)=p'(0)=\mu'(0)=R'(0)=0 \;.
\label{cond6}
\ee
Also from equation (\ref{G00}) we see that near the centre of the star
\be
m(r)\approx r^3\;.
\label{cond7}
\ee

\subsection{Conditions for thermodynamic stability}

From the matching conditions and regularity conditions given above we see that $c'(r_b)>0$ and $c'(0)=0$. Since the function $c(r)$ depends on the gravitational potential, there should be no local extremum of the potential at any non-central shell, for the thermodynamic stability of the star. Hence we impose, in the interior spacetime,
\be
c'(r)\ge0\;,
\label{cond8}
\ee
where the equality is only achieved at the central shell. Thus the function $c(r)$ will be a monotonically increasing function of the coordinate $r$, ranging from $c_0\equiv c(0)>0$ to $c(r_b)$ (which is given by (\ref{cond3})). Furthermore for physical stability of the stellar structure we also require that the total effective density of the star should be a non-increasing function of $r$. Therefore we must have
\be
\frac{d}{dr}\left(\frac{m}{r^3}\right)\le0\;.
\label{cond9}
\ee
And finally, to avoid the pathologies of $\fp=0$ anywhere inside the star (which may happen in the Starobinsky model and other viable models for negative Ricci scalar), we impose 
a further restriction on the Ricci scalar 
\be
R(r)\ge0\;.
\label{cond10}
\ee
Now we can easily see that $R_0\equiv R(0)\ge0$ while $R(r_b)=0$. For physically viable stellar models we would take the function $R(r)$ to be monotonic, which implies
\be
R'(r)\le0\;,
\label{cond11}
\ee
where the equality is only satisfied at the centre and at the surface of the star.

\section{Modification of Buchdahl-Bondi Limit}

In this section, we calculate the bound on the mass to radius ratio of a relativistic compact star, subject to the stability and regularity conditions (\ref{cond1})-(\ref{cond11}) of the previous section. As in GR, we shall see that this bound is independent of the equation of state of the standard matter. Let us denote
\be
\phi(r)=\frac{1}{r}\sqrt{1-\frac{2m}{r}}\;.
\ee
Then we can immediately see that
\be
\phi(r_b)=\frac{1}{r_b}\sqrt{1-\frac{2M}{r_b}}\;.
\ee
We can now prove easily the following proposition:
\begin{prop}
In the interior spacetime, for any shell with the radial coordinate `$r$', $\fp(R(r))c'(r)$ is positive and bounded from below.
\end{prop}
\begin{proof}
Using equation (\ref{cond9}) in (\ref{master2}) we get
\be
\fp\frac{d}{dr}\left[\phi(r)c'(r)\right]+c\frac{d}{dr}\left[\phi(r)(\fp)'\right]\le0\;.
\label{master3}
\ee
Integrating (by parts) the above equation from $r$ to $r_b$, where $r_b$ is the boundary of the star, we obtain
\bea
\left[\fp\left[\phi(r)c'\right]\right]_r^{r_b} +\left[c\left[ \phi(r)(\fp)'\right]\right]_r^{r_b}&&\nonumber\\
-2\int\limits_r^{r_b}\fpp R'c'\phi(r)&\le&0\;.
\eea
Since  in the interior spacetime $R'\le0$, the third term in LHS of the above equation has a positive contribution and hence dropping that term does not change (rather strengthens) the inequality. Therefore we have
\bea
\fp(R(r_b))\left[\phi(r_b)c'(r_b)\right]-\fp(R(r))\left[\phi(r)c'(r)\right]&&\nonumber\\
+c(r_b)\left[\phi(r_b)\fpp R'(r_b)\right]-c(r)\left[\phi(r)\fpp R'\right]&\le&0.\nonumber\\
\eea
Now using (\ref{cond5}) in the above equation, the third term vanishes, while by (\ref{cond11}) the fourth term has a positive contribution to the LHS and can be dropped without altering the inequality. Substituting $c'(r_b)$ from (\ref{cond4}) we get
\be
\fp(R(r))\frac{dc}{dr}\ge \frac{\fp(0)}{r_b^3}\frac{Mr}{\sqrt{1-\frac{2m}{r}}}\;.
\label{buch1}
\ee
\end{proof}
Since the above inequality is true $\forall r \in [0,r_b]$, we can integrate the above from the centre to the surface, without changing the inequality. Integrating we get
\bea
\fp(0)c(r_b)-\fp(R_0)c_0-\int\limits_0^{r_b}\fpp R'cdr&\ge&\nonumber\\
\frac{\fp(0)M}{r_b^3}\int\limits_0^{r_b}\frac{r}{\sqrt{1-\frac{2m}{r}}}dr.&&
\label{buch2}
\eea
To get a bound on the integral term above, we state and prove the following proposition:
\begin{prop}
For any thermodynamically stable stellar model, satisfying the conditions (\ref{cond1})-(\ref{cond11}), the integral of $(1/\phi(r))$ from the centre to the surface is bounded from below, and the bound is equal to $(r_b^3/2M)\left[1- r_b\phi(r_b)\right].$
\end{prop}
\begin{proof}
To prove this, we follow the same steps as given in \cite{stephani}. Let us write the total effective density function $\mu(r)$ as a sum of constant average density $\mu_0\equiv 6M/r_b^3$ and $\rho(r)$ as the variation on the average, that is
\be
\mu(r)=\frac{6M}{r_b^3}+\rho(r)\;.
\ee
This implies
\be
2m(r)=2M\frac{r^3}{r_b^3}+\int\limits_0^r\rho(r)r^2dr\;.
\label{buch3}
\ee
From the above equation we can easily see that
\be
\int\limits_0^{r_b}\rho(r)r^2dr=0\;.
\ee
Also since the effective density is a monotonically non increasing function from the centre to the surface of the star, we must have the central density greater than the average density. Therefore we have $\rho(0)\ge0$ and $\rho'\le0$. The integral term in (\ref{buch3}) is always positive and goes to zero as $r\rightarrow r_b$. This implies, $\forall r \in [0,r_b]$ that
\be
2M\frac{r^3}{r_b^3}\le 2m\;,
\ee
and therefore 
\be
\int\limits_0^{r_b}\frac{rdr}{\sqrt{1-\frac{2m}{r}}}\ge\int\limits_0^{r_b}\frac{rdr}{\sqrt{1-\frac{2M}{r}}}\equiv\frac{r_b^3}{2M}\left[1- r_b\phi(r_b)\right].
\label{buch4}
\ee
\end{proof}
We can use the above equation in the RHS of (\ref{buch2}), without altering the inequality. Hence we now get
\bea
\fp(0)c(r_b)-\fp(R_0)c_0-\int\limits_0^{r_b}\fpp R'cdr&\ge&\nonumber\\
\frac{\fp(0)}{2}\left[1- {\sqrt{1-\frac{2M}{r_b}}}\right].&&
\label{buch5}
\eea
Now consider the integral term in the LHS of the above equation. As $R'\le0$, we have
\be
-\int\limits_0^{r_b}\fpp R'c(r)d=\left|\int\limits_0^{r_b}\fpp R'c(r)dr\right|\;.
\ee
Since this term has a positive contribution in the LHS of (\ref{buch5}), we cannot drop it. However as we have seen $c'(r)\ge0$ throughout the interior of the star, therefore $c(r_b)\ge c(r), \forall r\in[0,r_b]$ and hence the inequality will not change if we substitute
\be
c(r_b)\left|\int\limits_0^{r_b}\fpp R'dr\right|\equiv c(r_b)\left(\fp(R_0)-\fp(0)\right),
\ee
for the integral term in (\ref{buch5}). On substitution, and using (\ref{cond3}), we finally have 
\be
\sqrt{1-\frac{2M}{r_b}}\left[\frac{\fp(0)}{2}+\fp(R_0)\right]-\frac{\fp(0)}{2}\ge\fp(R_0)c_0\;.
\label{buch6}
\ee
We are now in a position to state and prove the following theorem:
\begin{thm}
The regularity and thermodynamic stability conditions on a spherically symmetric and static star of radius $r_b$, immersed in the Schwarzschild vacuum in $f(R)$-gravity, impose an upper bound on the total effective mass of the star. This upper bound lies between the Buchdahl-Bondi bound for general relativity and the Schwarzschild static limit $2M=r_b$.
\end{thm}
\begin{proof}
From the regularity conditions at the centre of the star, and also by the conditions on the function $f(R)$, we can see that the RHS of the equation (\ref{buch6}) is strictly greater than zero. Therefore we get
 \be
\sqrt{1-\frac{2M}{r_b}}\left[\frac{\fp(0)}{2}+\fp(R_0)\right]-\frac{\fp(0)}{2}>0\;,
\ee
which, by squaring the LHS, can be simplified to the following expression
\be
\left[1+2\frac{\fp(R_0)}{\fp(0)}\right]^2\left(1-\frac{2M}{r_b}\right)>1\;.
\label{buch7}
\ee
Taking the quantity $2M$ to the RHS and simplifying we finally obtain the upper bound as
\be
2M< \frac{4\frac{\fp(R_0)}{\fp(0)}\left[1+\frac{\fp(R_0)}{\fp(0)}\right]}{\left[1+2\frac{\fp(R_0)}{\fp(0)}\right]^2}r_b\;.
\ee
By our assumptions on the Ricci Scalar and the function $f(R)$, we can immediately see that $\frac{\fp(R_0)}{\fp(0)}\ge1$. When this factor equals unity, we regain the usual Buchdahl-Bondi upper bound $2M<(8/9)r_b$. When this term is much larger than unity, this upper bound tends to the Schwarzschild static limit $2M=r_b$. 
\end{proof}
The above theorem ensures that if we suitably modify general relativity by a function $f(R)\ne R$, we can have a thermodynamically stable, regular spherical star with an effective mass $M$ in the region $(8/9)r_b\le 2M\le r_b$, which is the forbidden region in general relativity. Hence in these theories we can have more massive but stable compact stars, that may be one of the keys to the dark matter problem. Also for these stars the surface redshift $z\ge2$, which can in principle be detected experimentally. To see this effect more transparently, let us assume a small modification from general relativity. In other words let
\be
\frac{\fp(R_0)}{\fp(0)}=1+\alpha\;,\; 0\le \alpha <<1\;.
\ee
Then the upper bound on the effective mass of the star becomes
\be
2M<\frac89\left(1+\frac{\alpha}{6}\right)r_b\;.
\ee
Hence the net fractional increase of the effective mass, defined as $\delta M\equiv(M_{eff}- M_{GR})/M_{GR}$ (where $M_{GR}$ is the upper bound in general relativity), can be written as
\be
\frac{\delta M}{M_{GR}}=\frac{4\alpha}{27}\;.
\ee
This may easily and naturally account for the extra massive neutron stars in the sky. We can immediately calculate the maximum surface redshift for such stars. We get
\be
z<z_{max}\equiv 2(1+\alpha)
\ee
Hence, from our observational data from the compact stars, we can in principle experimentally verify any deviation from general relativity in the high curvature regime.

\section{Discussion}

In this paper, we studied model independent bounds on spherically symmetric stellar structures in $f(R)$-gravity. In particular, our results are independent of the matter distribution and the equation of state. Subject to very generic conditions of regularity and thermodynamic stability in the interior of the star and the matching conditions at the stellar surface to the vacuum Schwarzschild exterior, we transparently demonstrated that the mass to radius ratio of the star is bounded from above, and this is a stricter bound than the Schwarzschild static limit. In other words, we generalised the Buchdahl-Bondi bound on static stars in GR to $f(R)$-theories.

We also showed that this upper bound is larger than the Buchdahl-Bondi limit of GR, whenever $f(R)\ne R$. Hence, in principle, we can pack extra effective mass in a stable compact star in these theories, which is forbidden in GR.  These extra massive stars may be one of the solutions to the dark matter problem, that manifests through the rotation curves of the galaxies.

Furthermore, as a direct consequence of this extra effective mass, we proved that the surface redshift of an extra-massive compact star can be greater than 2, which is the upper limit in the case of general relativity. This gives a novel and interesting scenario, where we can observationally verify the validity or otherwise of general relativity in the strong gravity/ high curvature regime. 

An interesting possibility of further study in this scenario would be the violation of Chandrasekhar limit. We know, it is possible that Chandrasekhar limit may be violated in the presence of strong magnetic fields. The existence of {\em super-Chandrasekhar} white dwarfs has been motivated on the grounds of a polytropic equation of state in GR for an anisotropic matter distribution \cite{herrera}. It would be interesting to see, how the curvature terms in $f(R)$-gravity modify this limit.

\begin{center}
{\bf Acknowledgments}
\end{center}
AMN and RG are supported by the National Research Foundation (South Africa). SDM acknowledges that this work is based upon 
research supported by the South African Research Chair Initiative of the Department of Science and Technology.  


\end{document}